\documentclass[10pt,a4paper]{article}

\usepackage{amsfonts,amssymb,amsmath,amsthm}
\usepackage{fullpage}

\usepackage[english]{babel}
\usepackage[utf8]{inputenc}
\usepackage{algorithm}
\usepackage{tikz-cd}
\usepackage[noend]{algpseudocode}

\newtheorem{lemma}{Lemma}
\newtheorem{claim}{Claim}
\newtheorem{theorem}{Theorem}

\newcommand{\lref}[1]{(\ref{#1})}
\newcommand{\sset}[1]{\{ #1 \}}

\newcommand{\KM}[1]{{\color{blue} #1}}

\title{Ratio-Balanced Maximum Flows}

\author{Hannaneh Akrami\thanks{Sharif University of Technology, Iran} \and Kurt Mehlhorn\thanks{Max Planck Institute for Informatics, Germany} \and Tommy Odland\thanks{Sparebanken Vest, Norway}}

\begin{document}

\maketitle

\begin{abstract}When a loan is approved for a person or company, the bank is subject to \emph{credit risk}; the risk that the lender defaults. To mitigate this risk, a bank will require some form of \emph{security}, which will be collected if the lender defaults. Accounts can be secured by several securities and a security can be used for several accounts. The goal is to fractionally assign the securities to the accounts so as to balance the risk.

This situation can be modelled by a bipartite graph. We have a set $S$ of securities and a set $A$ of accounts. Each security has a \emph{value} $v_i$ and each account has an \emph{exposure} $e_j$. 
If a security $i$ can be used to secure an account $j$, we have an edge from $i$ to $j$. 
Let $f_{ij}$ be part of security $i$'s value used to secure account $j$. We are searching for a maximum flow that send at most $v_i$ units out of node $i \in S$ and at most $e_j$ units into node $j \in A$. Then $s_j = e_j - \sum_i f_{ij}$ is the unsecured part of account $j$. We are searching for the maximum flow that minimizes $\sum_j s_j^2/e_j$. 
\end{abstract}

\section{Introduction}

When a loan is approved for a person or company, the bank is subject to \emph{credit risk}; the risk that the lender defaults.
To mitigate this risk, a bank will require some form of \emph{security}, which will be collected if the lender defaults.
The bank opens a financial account for the loan, and one or more securities may be connected to it.
It is also possible that a security object is connected to more than one account.
Many-to-many relationships between securities and accounts rarely occur in the private market, but in the corporate market they are not uncommon. 

We can model this situation by a bipartite graph. 
We have a set $S$ of securities and a set $A$ of accounts. 
Each security has a \emph{value} $v_i$ and each account has an \emph{exposure} $e_j$. 
If a security $i$ can be used to secure an account $j$, we have an edge from $i$ to $j$. 
Let $E$ be the set of edges. 
The question is then how much of security $i$'s value should be used to secure account $j$.
Let us use $f_{ij}$ to denote this value. 
Clearly, we cannot use a security to more than its value and we do not want to secure an account to more than its amount, i.e.,\footnote{All summations (except if a summation range is explicitly specified) with summation index $i$ are over $i \in S$, all summations with summation index $j$ are over $j \in A$ and all summations with summation indices $i,j$ are over $ij \in E$.}
\begin{align}
  \label{v-constraint}  \sum_j f_{ij} &\le v_i \qquad  &\text{for all securities $i \in S$}\\
  \label{e-constraint} \sum_i f_{ij}  &\le e_j &\text{for all accounts $j \in A$.}
\end{align}
The unsecured part of the accounts is then $\sum_{j \in A}(e_j - \sum_i f_{ij}) = \sum_j e_j - \sum_{ij} f_{ij}$. 
Clearly, we want to make the unsecured part as small as possible, i.e., we want
\begin{equation}\label{max-flow}  \sum_{ij} f_{ij} \quad \text{to be maximum.} \end{equation}
In other words, we want a maximum flow from securities to accounts obeying the capacity constraints~\lref{v-constraint} and~\lref{e-constraint}. 

The \emph{surplus (unsecured part)} of an account $j$ is equal to $s_j = e_j - \sum_i f_{ij}$ and the \emph{unsecured fraction} or \emph{risk ratio} of an account $j$ is equal to $r_j = s_j/e_j$. 
It is desirable that all accounts are secured to the same fraction as much as possible. 
Formally, if security $i$ is used for account $j$ ($f_{ij} > 0$) and could be used for account $\ell$ ($i\ell \in E$), then the unsecured fraction of account $\ell$ is at most the unsecured fraction of account $j$ ($r_\ell \le r_j$). 
Otherwise, we could divert some of the flow $f_{ij}$ onto the edge $i\ell$ and make the secured fractions more equal. 
Formally,
\begin{equation}\label{ratio-balance} \text{if $f_{ij} > 0$ and $i\ell \in E$ then $r_j \ge r_\ell$.} \end{equation}

We have now defined the \emph{ratio-balanced maximum flow problem}: among the maximum flows satisfying the capacity constraints~\lref{v-constraint} and~\lref{e-constraint}, find the one that satisfies the ratio-constraint~\lref{ratio-balance}. 
The following example illustrates the concept. 

	

	
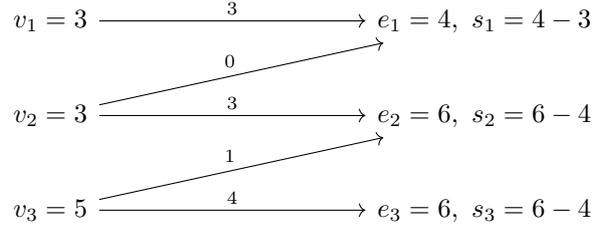
\[\begin{tikzcd}[column sep=10em, row sep=2em]
	v_1 = 3 \arrow{r}{3} & e_1 = 4,\ s_1 = 4 - 3 \\
	v_2 = 3 \arrow{r}{3} \arrow{ur}{0} & e_2 = 6,\ s_2 = 6 - 4 \\
	v_3 = 5 \arrow{r}{4} \arrow{ur}{1} & e_3 = 6,\ s_3 = 6 - 4 \\
	\end{tikzcd}\]
In the ratio-balanced maximum flow $f_{11} = 3 $, 
	$f_{21} = 0 $,
	$f_{22} = 3 $,
	$f_{32} = 1 $ and
	$f_{33} = 4 $, and the ratios are $r_1 = {1}/{4}$ and $r_2 = r_3 = {1}/{3}$. 
	
A related problem is to compute the flow that minimizes the squared 2-norm $\sum_j s_j^2$ of the unsecured parts. 
This problem is known as balanced flows~\cite{DPSV08} and can be solved in polynomial time. 
The papers~\cite{DPSV08,Darwish-Mehlhorn} suggested to us that ratio-balanced flows can be computed efficiently.

This paper is structured as follows. 
In Section~\ref{Alternative Characterization}, we give an alternative characterization for ratio-balanced maximum flows and show that they are the flows minimizing $\sum_j r_j^2 e_j$ subject to the capacity constraints~\lref{v-constraint} and~\lref{e-constraint}. 
In Section~\ref{Combinatorial Algorithm} we give a combinatorial algorithm and show that a ratio-balanced flow can be computed by at most $n \log (nM)$ maximum flow computation. 
This assumes that all values and exposures are integers bounded by $M$.
In Section~\ref{Quadratic Program} we give a quadratic program for ratio-balanced flows and in Section~\ref{Extensions} we discuss generalizations.

\section{Alternative Characterization}\label{Alternative Characterization}

We call a flow minimizing $\sum_j r_j^2 e_j$ an \emph{MWSR} (\emph{minimum weighted sum of squared risk ratios}) flow. 
Let $f$ and $g$ be two flows. 
We call $f$ and $g$ \emph{equivalent} if the risk ratios of all accounts with respect to $f$ and $g$ are equal, i.e., for all $j \in A$, $r_j^f = (e_j - \sum_i f_{ij})/e_j = (e_j - \sum_i g_{ij})/e_j = r_j^g$.

\begin{theorem} A flow $f$ is a ratio-balanced maximum flow if and only if it is MWSR. 
	All ratio-balanced flows are equivalent. 
\end{theorem}

\begin{proof} We first show that an MWSR flow is maximum and satisfies the ratio-constraint (4). 
	Thus an MWSR flow is ratio-balanced. 
	We then go on to show that any two ratio-balanced flows are equivalent.

\begin{claim} An MWSR flow subject to the capacity constraints is a maximum flow. 
\end{claim}

\begin{proof} Assume otherwise and let $f$ be an MWSR flow. 
If $f$ is not a maximum flow then there is an augmenting path with respect to it, i.e., a sequence $i_1,j_1,i_2,j_2,\ldots,i_k,j_k$ such that $i_\ell \in S$ and $j_\ell \in A$ for all $\ell$, $\sum_{j \in \delta(i_1)} f_{i_1j} < v_{i_1}$, $\sum_{i \in \delta(j_k)} f_{i j_k} < e_{j_k}$ and $f_{j_\ell i_{\ell + 1}} > 0$ for all $\ell$. 
We increase the flow on all edges $(i_\ell, j_\ell)$ by a small amount, decrease the flow on the edges $(j_\ell,i_{\ell + 1})$ by the same amount. 
We obtain a flow that obeys the capacity constraints and for which $r_{j_k}$ is smaller. 
\end{proof}

\begin{claim} A{n} MWSR flow subject to the capacity constraints satisfies the ratio-constraint~\lref{ratio-balance}. 
\end{claim}

\begin{proof}
The derivative of the objective with respect to $f_{ij}$ is equal to
  \[    -2 e_j r_j\frac{1}{e_j} = -2 r_j.\]
  Therefore decreasing the flow on $(i,j)$ by an infinitesimal amount $\varepsilon$ and increasing the flow on $(i,\ell)$ by the same amount, will change the objective by
  \[   (2r_j - 2r_\ell)\varepsilon = 2 (r_j - r_\ell) \varepsilon.\]
  If $r_j < r_\ell$, the change would be  negative, a contradiction.
\end{proof}

We have shown that a MWSR flow is ratio-balanced. 
Now we prove that all ratio-balanced flows are equivalent. 
We may assume that every security node can be used for some account. 
Otherwise, we may simply remove the security. 
Let $f$ and $g$ be two ratio-balanced flows. 

Our proof is by induction on the number of nodes in $S$. 
If $|S| = 0$ then $f$ and $g$ are equivalent.
Assume $|S| > 0$ and for every graph in which the number of security nodes is less than $|S|$, $f$ and $g$ are equivalent. 
For any $j \in A$, let $r^f_j$ and $r^g_j$ be the risk-ratio of node $j$ under $f$ and $g$ respectively. 
For any $j \in A$ and $i \in S$, let $f_{ij}$ and $g_{ij}$ be the flow from $i$ to $j$ under $f$ and $g$ respectively. 

Without loss of generality we assume that the maximum risk ratio under the flow f is no smaller than the maximum risk ratio under the flow $g$, i.e., $R := \max_j r^f_j \geq \max_j r^g_j $.

If $R = 0$ then $r^f_j = r^g_j = 0$ for all $j \in A$ and $f$ and $g$ are equivalent. 

Now assume that $R > 0$. 
Let $A' = \{ j; r^f_j = R\}$ be the least secured nodes under $f$. 
Let $S'$ be the set of nodes is $S$ which send positive flow to nodes in $A'$ under $f$. 
Since $f$ is ratio-balanced, there is no edge from $S \backslash S'$ to nodes in $A'$ and $f_{ij} = 0$ for $i \in S'$ and $j \in A \setminus A'$. 
Moreover, since any security node $i \in S'$ is connected to a $j$ such that $r^f_j = R > 0$, $\sum_j f_{ij} = v_i$. 
Otherwise, more flow can be sent through $ij$ contradicting $f$ being a maximum flow.

With respect to $f$, the total outflow of the nodes in $S'$ is equal to the total inflow of the nodes in $A'$:
$$ \sum_{i \in S'} v_i = (1 - R) \sum_{j \in A'} e_j.$$     
With respect to $g$, the total inflow of the nodes in $A'$ is at most the total outflow of the nodes in $S'$ (there might be flow from $S'$ to $A \setminus A'$):
$$\sum_{j \in A'} (1 - r^g_j) e_j \leq \sum_{i \in S'} v_i.$$ 
Therefore,
\begin{align*}(1 - R) \sum_{j \in A'} e_j &\geq \sum_{j \in A'} (1 - r^g_j) e_j\\
  \intertext{and hence}
  \sum_{j \in A'} r^g_j \cdot e_j &\geq \sum_{j \in A'} R \cdot e_j.
\end{align*}
By definition of $R$ we have $R \geq r^g_j$ that for all $j \in A$. 
So for every $j \in A'$, $r^g_j = R$ and also
$$\sum_{j \in A'} (1 - r^g_j) e_j = \sum_{i \in S'} v_i,$$
which means that also in $g$, all flow from $S'$ goes into $A'$. 

Now remove $A' \cup S'$ from the graph. 
The number of security nodes is reduced and according to the induction assumption, $f$ and $g$ are equivalent in the reduced graph.
\end{proof}

\section{Combinatorial Algorithm}\label{Combinatorial Algorithm}

We now give the algorithm for computing a ratio-balanced flow $f$. 
The algorithm works in phases. 
In each phase, it finds a maximum flow and subsets of $S$ and $A$. 
We denote the flow determined in the $k$-th phase by $f^{(k)}$ and the subsets by $S'_k$ and $A'_k$. 
The flow $f$ agrees with $f^{(k)}$ on all edges from $S'_k$ to $A'_k$ and has flow zero on all edges from $S'_k$ to $A \setminus \cup_{i \le k} A'_i$. 

Let $\lambda$ be a rational number in $[0,1]$ and consider the following flow problem $P_\lambda$. 
We add a source node $s$ and an edge $(s,i)$ of capacity $v_i$ for every $i \in S$. 
We add a sink node $t$ and an edge $(j,t)$ of capacity $\lambda e_j$ for every $j \in A$. 
We set the capacity to $+\infty$ for all edges from $S$ to $A$. 

Let $\lambda_1$ be the maximum $\lambda$ such that $P_\lambda$ has a feasible solution. 
We discuss below how to find $\lambda_1$. 
Consider the residual network with respect to the maximum flow $f^{(1)}$ in $P_{\lambda_1}$ and let $S'_1$ and $A'_1$ be the nodes that cannot be reached from $s$ by a path in the residual network.
Remove $S'_1$ and $A'_1$ from the graph and recurse until either $S$ or $A$ is empty. 

\begin{theorem}
The flow $f$  is a maximum ratio-balanced flow.
\end{theorem}

\begin{proof}

  Let $S_k$ and $A_k$ be the set of remaining securities and accounts in the beginning of $k$-th step respectively; $S_1 = S$ and $A_1 = A$. 
  In the $k$-th phase, the flow network has vertices $S \setminus \cup_{i < k} S_i'$ on the $S$-side and vertices $A \setminus \cup_{i < k} A_i'$ on the $A$-side. 
  Let $\lambda_k$ be the maximum $\lambda$ such that $P_\lambda$ has a feasible solution in the $k$-th phase, let $f^{k)}$ be the maximum flow in the $k$-th phase and let $S'_k \in S_k$ and $A'_k \in A_k$ be the nodes that cannot be reached from $s$ by a path in the residual network with respect to $f^{(k)}$. 

Clearly $f^{(k)}_{ij} = 0$ for $i \in S'_k$ and $j \in A_k \setminus A'_k$ and $(i,j) 
\not\in E$ for $i \in S_k \setminus S'_k$ and $j \in A'_k$ because of non-reachability  in the residual network.

Also for $j \in A_k \setminus A'_k$, the security ratio is larger than $\lambda _k$
because the residual network certifies that we can send more flow. 
Which means that if $\ell > k$, $\lambda_{\ell} > \lambda_k$. 

So if $f_{ij} > 0$ then there exists $k$ such that $i \in S'_k$ and $j \in A'_k$. 
There is no edge from $i$ to $S'_{\ell}$ such that $\ell < k$ and for any $j \in S'_{\ell}$ such that $\ell > k$, $f_{ij} = 0$. 
Therefore, $f$ satisfies the ratio-constraint~\lref{ratio-balance}.

Next we prove that $f$ is a maximum flow.

If $\lambda_k < 1$, then for all $i \in S'_k$, $\sum_j f_{ij} = v_i$. 
Otherwise $i$ would be reachable from $s$. 

Let $h$ be the number of steps. 
If $\lambda_h < 1$, then the total flow under $f$ is $\sum_{i \in S} v_i$. 
The total flow in a maximum flow cannot exceed this amount. 
So, $f$ is a maximum flow.

Now assume that $\lambda_h = 1$. 
It means for every $j \in A'_h$, $\sum_i f_{ij} = e_j$. 
Hence, the total flow in $f$ is 
$$\sum_{i \notin S'_h} v_i + \sum_{j \in A'_h} e_j.$$
Now let the total flow in a maximum flow be $F$. 
Then
$$ F = \sum_{i \in S} \sum_{j \in A} F_{ij} = \sum_{i \notin S'_h} \sum_{j \in A} F_{ij} + \sum_{i \in S'_h} \sum_{j \in A} F_{ij} \leq \sum_{i \notin S'_h} v_i + \sum_{j \in A'_h} e_j. $$

\noindent
The last inequality holds because there is no edge from $S'_h$ to $A \backslash A'_h$. So any outflow from $S'_h$ is inflow for $A'_h$. 
Therefore, $f$ is a maximum flow.
\end{proof}

We next show how to find $\lambda_i$ efficiently. 
For this we assume that the $v_i$ and $e_j$ are integers and use $M$ to the denote their maximum. 
For any $j \in A$, let $r^f_j$ be the risk-ratio of node $j$ under $f$. 

\begin{lemma}
For every $j \in A$, $r^f_j$ is a rational number with numerator and denominator bounded by $nM$. 
\end{lemma}

\begin{proof}
If $r^f_j = 0$, the claim is obvious. 
So assume $r^f_j > 0$. 
Consider the $k$ such that $j \in A'_k$. 
All flow into the nodes in $A'_k$ comes from the nodes in $S'_k$ and the total flow from $S'_k$ to $A'_k$ is equal to $\sum_{i \in S'_k} v_i$. 
All nodes in $A'_k$ have the same risk-ratio. 
This ratio is equal to $1 - \frac{\sum_{i \in S'_k} v_i} {\sum_{\ell \in A'_k} e_\ell}$.
\end{proof}

\begin{theorem} \cite{Papadimitriou79}
Let $x$ be a fraction, both numerator and denominator of which are bounded by $M$. 
Then $x$ can be determined by $O(log(M))$ queries of form "is $x \leq p/q$?", where $p, q \leq 2M$, and $O(log (M))$ arithmetic operations on integers of size not greater than $2M$. 
\end{theorem}

Instead of finding $\lambda_i$, we find $1 - \lambda_i$ which is also a fraction with both numerator and denominator bounded by $nM$. 
In order to answer each query, we check if $P_{1 - p/q}$ has a feasible solution or not. 
If it does, then $\lambda_i \geq 1 - p/q$ which means $1 - \lambda_i \leq p/q$. 
Otherwise, $\lambda_i > 1 - p/q$ or $1 - \lambda_i > p/q$.

We need to find at most $n$ $\lambda$-values. 
For each one we need to answer $\log (nM)$ queries. 
Each query is a maxflow-computation.

\begin{theorem} Let the $v_i$'s and $e_j$'s be integer and let $M$ be their maximum. 
A ratio-balanced flow can be computed with $n \log (nM)$ maxflow-computations. 
\end{theorem}

For balanced flows (definition given in the introduction), the number of maxflow-computations can be reduced to a single parameterized flow computation~\cite{Darwish-Mehlhorn}. 
The same improvement might be possible here.

\section{Solution by Formulation as a Quadratic Program} \label{Quadratic Program}

The task ``minimize $\sum_j e_j r_j^2$ subject to~\lref{v-constraint} and~\lref{e-constraint}'' is a quadratic program. 
As such it can be (approximately) solved by any QP-package, e.g., CVXOPT\cite{cvxopt}. 
For concreteness, we give the formulation as a standard 
QP problem in the notation used in CVXOPT. 
\begin{align}
\nonumber \underset{\mathbf{x}}{\operatorname{minimize}} \quad  & \frac{1}{2} \mathbf{x}^T P \mathbf{x} + \mathbf{q}^T \mathbf{x} \\
\label{eq:QP} \text{subject to } \quad & G\mathbf{x}  \leq \mathbf{h} \\ 
\nonumber
& A\mathbf{x} = \mathbf{b}
\end{align}
We use $\mathbf{1}$ for the all-ones column vector and $\hat{\mathbf{e}}_j$ for the $j$-th unit vector. 
We number the edges of the graph arbitrarily and use $\mathbf{x}$ for the vector of flows. 
The matrices $K$ and $V$ connect the flow variables to the securities and accounts:
\[ K_{ij} = 
	\begin{cases} 
	1 & \text{if } x_j \text{ is incident to } e_i \\
	0       & \text{else}
	\end{cases}  \qquad
	V_{ij} = 
\begin{cases} 
1 & \text{if } x_j \text{ is incident to } v_i \\
0       & \text{else}
\end{cases}\]
Figure~\ref{ex:large_general} shows an example. 
We are now ready to formulate the objective function and the constraints as a QP.

\begin{figure}[t]
  \begin{equation*} \footnotesize
    \begin{tikzcd}[column sep=6em, row sep=1.5em]
v_1 = 8 \arrow{dr}{x_{2}} \arrow{r}{x_{1}}  & e_1 = 12 \\
& e_2 = 8 \\
v_2 = 8 \arrow{r}{x_{5}} \arrow{ur}{x_{4}} \arrow{uur}{x_{3}} & e_3 =16
\end{tikzcd}\quad
\mathbf{x} = \begin{pmatrix}
x_{1} \\ 
x_{2} \\ 
x_{3} \\ 
x_{4} \\ 
x_{5} 
\end{pmatrix} \quad
K = \begin{pmatrix}
1 &  0& 1 & 0 &0 \\ 
0& 1 & 0 & 1 &0 \\ 
0& 0 & 0 & 0 & 1
\end{pmatrix} \quad
V = \begin{pmatrix}
1 &  1& 0 & 0 &0 \\ 
0& 0 & 1 & 1 & 1
\end{pmatrix}
\end{equation*}
\caption{\label{ex:large_general}  An example with three securities and two accounts.}
\end{figure}
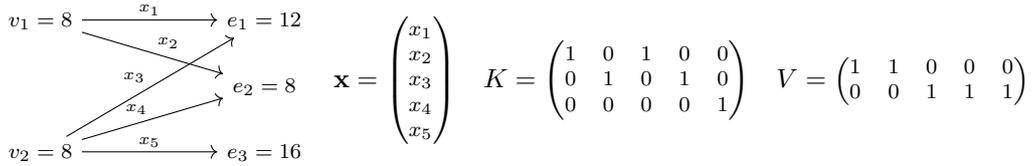

\begin{lemma} In standard form, the exposure-weighted sum-of-squares error function can be written as follows:
	\begin{equation*}
	\sum_{j }  e_j r_j^2 
	=
		\sum_{j }  
		\frac{\left(e_j - \hat{\mathbf{e}}_j^T K \mathbf{x} \right)^2}{e_j} = \sum_{j } e_j +
		\frac{1}{2} \mathbf{x}^T \left( 2 K^T \operatorname{diag} \left( (e_j)^{-1} \right) K \right) \mathbf{x} 
		- 2\cdot \mathbf{1}^T \mathbf{x},
	\end{equation*}
	so that $P = 2 K^T \operatorname{diag}\left( (e_j)^{-1} \right) K$ and $\mathbf{q} = -2 \cdot \mathbf{1}$ in Equation \eqref{eq:QP}.
\end{lemma}
\begin{proof}
	The goal is to write $\sum_j e_j r_j^2$ in the form $\frac{1}{2} \mathbf{x}^T P \mathbf{x} + \mathbf{q}^T \mathbf{x}$.
	To do so, we expand the square by writing
	\begin{align*}
	\sum_{j} ( e_j)^{-1} \left( e_j - \hat{\mathbf{e}}_j^T K \mathbf{x} \right)^2 
	&= \sum_{j }  ( e_j)^{-1} \left( e_j^2 - 2  e_j^T \hat{\mathbf{e}}_j^T K \mathbf{x}  +  \mathbf{x}^T K^T \hat{\mathbf{e}}_j \hat{\mathbf{e}}_j^T K \mathbf{x} \right) \\
	&= \sum_{j }   e_j - \sum_{j =1}^n  2  \hat{\mathbf{e}}_j^T K \mathbf{x}  + \mathbf{x}^T K^T \left( \sum_{j } ( e_j)^{-1} \hat{\mathbf{e}}_j \hat{\mathbf{e}}_j^T \right ) K \mathbf{x} \\
	&=    \underbrace{- 2 \mathbf{1}^T}_{ \mathbf{q}^T}  \mathbf{x}  + \mathbf{x}^T 
	\underbrace{K^T \operatorname{diag} \left( ( e_j)^{-1} \right) K}_{\frac{1}{2}P} \mathbf{x} + \sum_{j = 1}^n e_j.
	\end{align*}
	We 
        used the fact that $\mathbf{1}^T K = \mathbf{1}^T$ since $\sum_{i} K_{i j} = 1$ for every column $j$.
	The diagonal matrix $\operatorname{diag}(e_j^{-1})$ has $e_j^{-1}$ in position $(j, j)$. 
\end{proof}
In matrix notation, the constraints are $-I \mathbf{x} \le 0$ and $V \mathbf{x} \le \mathbf{v}$ and $K \mathbf{x} \le \mathbf{e}$. 

    For the example in Figure~\ref{ex:large_general}, the CVXOPT package solves the QP in 10 milliseconds
    and the algorithm uses 5 iterations. 
    The reported solution is 
\begin{equation*}
	\mathbf{x} = (x_1, x_2, x_3, x_4, x_5) = \left( 4.88, 3.12, 0.46, 0.43, 7.11 \right),
\end{equation*}
which gives (almost) equal risk ratios
\begin{equation*}
	\frac{4.88 + 0.46}{12} \approx \frac{3.12 + 0.43}{8} \approx \frac{7.11}{16} \approx 0.444.
\end{equation*}


\section{Extensions}\label{Extensions}
We discuss some extensions.

\paragraph{Over-Coverage:} Some accounts will be fully covered, meaning that their $r_j$'s will be zero.
Let $A'$ be the set of accounts that are fully covered and let $S'$ be the securities sending flow to them. We restrict the flow problem $P_\lambda$ to these accounts and securities and then proceed as in Section~\ref{Combinatorial Algorithm}. Let $\lambda_1$ be the maximum $\lambda \ge 1$ such that $P_\lambda$ has a feasible solution. 
Consider the residual network with respect to the maximum flow $f^{(1)}$ in $P_{\lambda_1}$ and let $S'_1$ and $A'_1$ be the nodes that cannot be reached from $s$ by a path in the residual network. 
Remove $S'_1$ and $A'_1$ from the graph and recurse until either $S'$ and $A'$ are empty. In the example below, $\lambda_1 = 1$, $S'_1 = \sset{1}$, $A'_1 = \sset{1}$, $f_{11} = 1$ and $f_{12} = 0$. Next, we have $\lambda_2 = 5$, $S'_2 = \sset{2,3}$, $A'_2 = {2}$, and $f_{22} = 2$ and $f_{32} = 3$. 

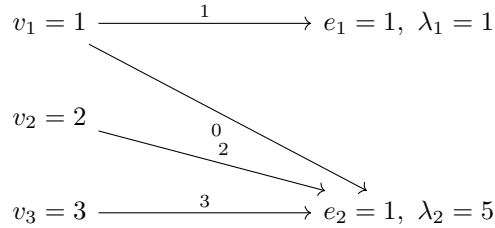
\begin{center}
  \begin{tikzcd}[column sep=8em, row sep=2em]
	v_1 = 1 \arrow[swap]{ddr}{0} \arrow{r}{1}  & e_1 = 1,\ \lambda_1 = 1 \\
	v_2 = 2  \arrow{dr}{2} &  \\
	v_3 = 3  \arrow{r}{3} & e_2 = 1,\ \lambda_2 = 5
      \end{tikzcd}
      \end{center}

\paragraph{Limits to a Claim:}
An account might contractually only have claim to parts of the security value. 
This is easily modeled by introducing an upper bound on the flow from a security to an account. 
The QP-algorithm and the combinatorial algorithm can handle such bounds.

\paragraph{Priorities:}
In the real world, accounts are often arranged by their \emph{priority} to a security object. 
If two accounts have prioritized claims to a security object, the account with highest priority (lowest priority number) gets its demand covered first. 
Remaining value goes to lower priority accounts. 
In the following example, we use 
parenthesized superscripts to denote priorities; account 2 has a lower priority (higher priority number) than account 3 on security 2. 
In other words, account 3 has ``first rights.''
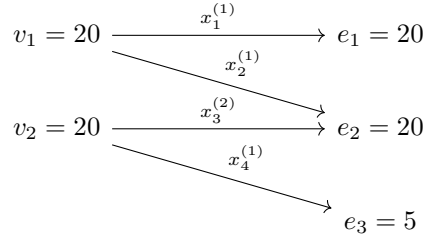
	\begin{equation*}
	\begin{tikzcd}[column sep=8em, row sep=2em]
	v_1 = 20 \arrow{dr}{x_{2}^{(1)}} \arrow{r}{x_{1}^{(1)}}  & e_1 = 20 \\
	v_2 = 20 \arrow{r}{x_{3}^{(2)}} \arrow{dr}{x_{4}^{(1)}}  & e_2 = 20 \\
	& e_3 = 5
	\end{tikzcd}
      \end{equation*}
      Assume we have $P$ different priority classes. 
      The desired solution is a maximum flow on the edges of priority 1. 
      Subject to this, it should be a maximum flow on the edges of priority 2, and so on. 
      Subject to this, it should be a maximum flow on the edges of priority $P$. 
      Subject to this, the flow should be ratio-balanced.

In the example above, the desired solution is $\mathbf{x} =  (17.5, 2.5, 15, 5)$. 
The flow on the edges of priority 1 is 25 and the flow on the edges of priority 2 is 15. 
Subject to this the flow balances the uncovered fractions of accounts 1 and 2. 

Frederic Dorn (Sparebanken Vest) suggested the use of minimum cost flows for modeling the priorities in combination with the objective for a balanced flow. 
Let $p=1,2,\ldots, P$ be the priorities, let $E^{{(p)}}$ be the set of edges with priority $p$ and $\epsilon$ a small number. 
Consider now the following optimization problem, which is a QP.
\begin{equation} \label{eqn:max_flow_quad} \underset{f}{\operatorname{minimize}} \quad 
- \sum_{p=1}^{P} \epsilon^p
\sum_{(i,j) \in E^{(p)}}
f_{ij}
+
\epsilon^{P+1} \sum_{j }   e_j r_j^2 \end{equation}
subject to $f \ge 0$ and \lref{v-constraint} and~\lref{e-constraint}. 
The first term in the objective sends as much flow as possible through the graph, but prioritizing first priority much stronger than second, the second much stronger than third, and so forth.
The second term states that everything else being equal, a ratio-balanced flow is preferable.

The combinatorial approach of Section~\ref{Combinatorial Algorithm} works too. 
We only have to replace the use of a maximum flow algorithm by the use of a minimum cost flow algorithm which minimizes the linear part of the objective in~\lref{eqn:max_flow_quad}.

\newcommand{\htmladdnormallink}[2]{#1}

\bibliographystyle{alpha}

\bibliography{ref,KMQuadratic}   

\end{document}